\documentclass{article}
\usepackage{authblk}
\usepackage{amsmath}
\usepackage{amsthm}
\usepackage{amsfonts}
\usepackage{hyperref}
\usepackage[vcentermath]{youngtab}
\usepackage[linesnumbered]{algorithm2e}
\usepackage{color}
\usepackage{graphicx}
\usepackage{braket}
\allowdisplaybreaks

\newtheorem{theorem}{Theorem}

\DeclareMathOperator{\poly}{poly}
\DeclareMathOperator{\per}{per}

\DeclareGraphicsExtensions{.pdf}

\def\Id{1\!\mathrm{l}}
\newcommand{\mr}[1]{\mathrm{#1}}

\begin{document}
\title{Quantum simulation of partially distinguishable boson sampling}
\author[1,2,3]{Alexandra E. Moylett\thanks{alex.moylett@bristol.ac.uk}}
\author[1]{Peter S. Turner\thanks{peter.turner@bristol.ac.uk}}
\affil[1]{Quantum Engineering Technology Labs, H. H. Wills Physics Laboratory and Department of Electrical \& Electronic Engineering, University of Bristol, BS8 1FD, UK}
\affil[2]{Quantum Engineering Centre for Doctoral Training, H. H. Wills Physics Laboratory and Department of Electrical \& Electronic Engineering, University of Bristol, BS8 1FD, UK}
\affil[3]{Heilbronn Institute for Mathematical Research, University of Bristol, BS8 1SN, UK}

\date{\today}

\maketitle

\begin{abstract}
Boson Sampling is the problem of sampling from the same output probability distribution as a collection of indistinguishable single photons input into a linear interferometer. 
It has been shown that, subject to certain computational complexity conjectures, in general the problem is difficult to solve classically, motivating optical experiments aimed at demonstrating quantum computational ``supremacy''.
There are a number of challenges faced by such experiments, including the generation of indistinguishable single photons.
We provide a quantum circuit that simulates bosonic sampling with arbitrarily distinguishable particles.
This makes clear how distinguishabililty leads to decoherence in the standard quantum circuit model, allowing insight to be gained.
At the heart of the circuit is the quantum Schur transform, which follows from a representation theoretic approach to the physics of distinguishable particles in first quantisation.
The techniques are quite general and have application beyond boson sampling.
%to exactly sample from this distribution on a classical computer must take superpolynomial time unless $P^{\#P} = BPP^{NP}$ and the Polynomial Hierarchy collapses to the third level. 
\end{abstract}

\section{Introduction}
\label{sec:intro}

Ever since the idea of quantum computers was first proposed, there has been significant interest in demonstrating how much more powerful they are than their classical counterparts. 
This has been shown for a number of problems, from factoring integers~\cite{shor1997} to solving linear equations~\cite{harrow2009}. 
However, these algorithms require large fault-tolerant quantum computers, which have proven challenging to develop. 
This has led to the search for a so called quantum advantage: a quantum experiment can be demonstrated in a laboratory in the near future, yet is hard to simulate on a classical computer~\cite{harrow2017}. 
Numerous proposals have been made~\cite{bremner2016, boixo2016}, with perhaps the best known being Boson Sampling~\cite{aaronson2011}, which considers the complexity of sampling from the same probability distribution as $n$ indistinguishable photons output from an $m$-mode linear optical interferometer, shown to be determined by matrix permanents~\cite{scheel2008}.

In their breakthrough article, Aaronson and Arkhipov showed that the ability to exactly simulate Boson Sampling in polynomial time would imply that $P^{\#P} = BPP^{NP}$ and the Polynomial Hierarchy would collapse to the third level. 
The same result was also shown for approximately sampling from the same distribution, modulo two conjectures related to the permanents of Gaussian matrices~\cite{aaronson2011}. 
This has led to the development of numerous experiments, with demonstrations as large as five and six photons~\cite{carolan2015, wang2017}.

However, going from these small demonstrations to being able to significantly outperform classical simulation is not a simple task, as recently highlighted by two papers. 
The first, by Neville et al.~\cite{neville2017}, gives an empirical evaluation of approximately sampling from the boson distribution when there are no collisions via a Markov Chain Monte Carlo algorithm, showing that they could produce 250 samples for 30 bosons across 900 modes in under five hours. 
The second, by Clifford and Clifford~\cite{clifford2017}, gives the first algorithm for exact Boson Sampling that is more efficient than computing the entire probability distribution, sampling from $n$ photons across $m$ modes in $O(n2^n + \poly(n,m))$ time and $O(m)$ space.

Any claim of a quantum advantage will furthermore need to take into account the practical challenges associated with Boson Sampling. 
One such challenge, which will be the focus of this work, is the need for indistinguishable single bosons. 
It is also clear that the particles being indistinguishable is to some extent a requirement for the problem to remain computationally hard, as there are known algorithms when the bosons are perfectly distinguishable~\cite{aaronson2014}.
Generating large numbers of indistinguishable photons is infeasible with current techniques, with the probability going like $p^n$ where $p$ is the probability of generating a single photon~\cite{gard2015}. 

In recent years there has been a wealth of theory on linear optical interferometry with partially distinguishable photons~\cite{rohde2015, shchesnovich2015, tichy2015, tillmann2015, menssen2017, laibacher2017}, which has lead to some known results on the computational complexity of sampling. 
Renema et al.~\cite{renema2017} give a classical algorithm for approximate sampling, which they use as a lower bound for how indistinguishable the photons must be in order to achieve a significant quantum advantage. 
Shchesnovich used average mutual fidelity to give an upper bound on the problem's complexity, stating that a sampling experiment is more powerful when the single-photon mode mismatch scales as $O(n^{-3/2})$ for $n$ photons~\cite{shchesnovich2014}. 
Rohde and Ralph~\cite{rohde2012} also briefly discuss this problem, using narrowband filtering to relate distinguishability to loss in order to give an upper bound.

In this paper we explore the problem of sampling from a collection of $n$ partially distinguishable single bosons interacting on a $m$-mode interferometer from the opposite direction, that is, from the perspective of quantum simulation. 
We do so by observing that ideal Boson Sampling is equivalent to sampling from the totally symmetric representation of the unitary group, and that partial distinguishability generalises this to the problem of sampling from arbitrary irreducible representations.
We show that quantum circuits for the Schur transform~\cite{bacon2007} can therefore be used to give a polynomial time quantum algorithm for sampling from the same probability distribution as bosons emerging from a linear interferometer, regardless of distinguishability.

Although it is generally accepted that quantum computers can perform Boson Sampling efficiently, there are few places where such algorithms are actually described explicitly. 
An example of such a method for the ideal (indistinguishable) case is by Aaronson and Arkhipov~\cite{aaronson2011}, using a technique by Reck et al.~\cite{reck1994} to decompose the unitary matrix $U$ into a sequence of $O(d^2)$ optical elements, each implemented via the Solovay-Kitaev theorem~\cite{nielsen2010}. 
%An example for fermions is given by Abrams and Lloyd~\cite{abrams1997}. 
%An interesting next step would be to try and get more of an idea of the runtimes for these two quantum algorithms, and see how they perform against one another. 
%As it currently stands, the only algorithm of these three with an explicit runtime calculated is the one by Abrams and Lloyd, which runs in $O(n^2\log^2d)$ time. 
%In comparison the algorithm presented here and the algorithm by Aaronson and Arkhipov both only have their runtimes described as polynomial time.
Here we show explicitly how nonideal linear optics can be viewed as a quantum computation, allowing a wider range of realistic experimental situations to be considered.
Our approach shows that while the ideal case is intimately related to the symmetric representation of the unitary group through matrix permanents, in the nonideal case all representations play a role.

The paper is structured as follows. 
In Section \ref{sec:preliminary-notes} we give an overview of background material including Boson Sampling, irreducible representations of the unitary and symmetric groups, and Schur-Weyl duality. 
We provide a simple quantum circuit for ideal indistinguishable photon sampling in Section \ref{sec:qc-bs}, before introducing the full circuit for sampling with distinguishable photons in Section \ref{sec:qc-noisy}. 
In Section \ref{sec:uu-duality}, we provide some further explanation as to why these circuits work, via what is known as unitary-unitary duality.
Following this result, we discuss a few interesting consequences: in Section \ref{sec:post-bs}, we show how postselection can be used with this circuit to sample from the indistinguishable distribution when given a distinguishable input; in Section \ref{sec:loss}, we consider how the circuit can be used to simulate Boson Sampling when photons are lost; and in Section \ref{sec:mixture}, we consider the multipartite entanglement of the output in the distinguishable case. 
%We conclude with some discussion and future work in Section \ref{sec:conclusion}.

\section{Preliminaries}
\label{sec:preliminary-notes}

\subsection{Sampling from bosonic distributions}

We start by defining the ideal probability distribution of indistinguishable single bosons interacting on a linear interferometer.
We'll refer to this as bosonic sampling, as it's a bit more general than Aaronson and Arkhipov's Boson Sampling problem as we describe below.
The input is $U \in \mathrm{U}(m)$, an $m\times m$ unitary matrix which describes an $m$-mode linear interferometer, and $S = (S_1,S_2,\dots,S_m)$ with $\sum_{i=1}^m S_i =n$, an ordered list of integers that corresponds to an $n$-boson, $m$-mode occupation describing the input state with $S_i$ bosons in mode $i$. 
Given an output occupation $S'$, define the $n \times n$ (not necessarily unitary) matrix $U_{S',S}$ as that formed by first taking $S_i'$ copies of row $i$ of $U$ in order to create an $m\times n$ matrix, from which we then take $S_j$ copies of column $j$. 
We can then define $\mathcal{D}_{U,S}$, the probability distribution for measuring an $n$-boson $m$-mode occupation $S'$ for interferometer $U$ and input state $S$, as
\begin{equation}\label{eqn:bs-distribution}
\textrm{Pr}_{\mathcal{D}_{U,S}}[S'] = \frac{|\per(U_{S',S})|^2}{\prod_{i=1}^m S_i'! S_i!} ,
\end{equation}
where $\per$ is the matrix permanent.

In a photonics experiment, this setting is described in terms of creation operators $a^\dag_i$ for a photon in mode $i$. 
The initial state is then
\begin{equation}
|S\rangle = \prod_{i=0}^m \frac{(a_i^\dagger)^{S_i}}{\sqrt{S_i}}|0^n\rangle.
\end{equation}
The evolution of the photonic state induced by a linear optical interferometer implementing $U$ can then be expressed as $a_i^\dagger \mapsto \sum_{j = 0}^m U_{i,j}a_j^\dagger$.
Thus single boson states evolve under linear interferometry just as a $m$ dimensional qudit does under a unitary gate $U$ (sometimes called unary encoding).
This suggests how quantum circuits simulating photonics might be constructed, as we'll see.

The problem known as Boson Sampling is that of sampling from this probability distribution in the case where the input occupation is specified as $|1^n 0^{n^2 - n}\rangle = \prod_{i = 1}^n a_{i}^\dagger|0\rangle$, and $U$ is drawn Haar randomly from U$(m=n^2)$\cite{aaronson2011}. 
It was proven by Aaronson and Arkhipov that if there was a polynomial time classical algorithm for sampling from this distribution, then $P^{\#P} = BPP^{NP}$ and the Polynomial Hierarchy would collapse to the third level. 
They also showed similar results for approximate Boson Sampling, where samples are drawn from a distribution $\mathcal{D}_{\mathcal{O}(U, \epsilon)}$ such that $||\mathcal{D}_{\mathcal{O}(U, \epsilon)} - \mathcal{D}_U|| \leq \epsilon$ for all $U \in \textrm{U}(m)$. 
%Their proof showed that if someone could perform approximate Boson Sampling on a classical computer in polynomial time, and the Permanent-of-Gaussians and Permanent Anti-Concentration conjectures hold, then again the Polynomial Hierarchy would collapse to the third level~\cite{aaronson2011}.

\subsection{Schur-Weyl duality}

Our algorithm can be understood from the perspective of the representation theory of the unitary group U$(m)$ of linear interferometers acting of $m$ modes.
The irreducible representations, or irreps, are intimately related to those of the symmetric group $\textrm{S}_n$ that permutes the particles.
Irreps of both of these groups are indexed by ordered partitions $\lambda = (\lambda_1,\lambda_2,\cdots,\lambda_m)$ of $n$ such that $\lambda_i \geq \lambda_{i+1}$ and $\sum_{i = 1}^m \lambda_i = n$.
We usually suppress zeros in this notation, so for example the totally symmetric irrep $\lambda=(n, 0,\cdots,0)$ is written $(n)$. 
The number of nonzero $\lambda_i$ is called the length of the partition, $\ell(\lambda)$, and only partitions with $\ell(\lambda) \leq m$ occur, which we will assume in all of our expressions that follow.
%Following the convention of~\cite{rowe2012}, we shall label irreps of S$_n$ and U$(m)$ indexed by $\lambda$ as $(\lambda)$ and $\{\lambda\}$, respectively.

For the tensor space of $n$ $m$-dimensional qudits, the actions of the symmetric and unitary groups on a state $|\Psi\rangle \in (\mathbb{C}^m)^{\otimes n}$ are explicitly as follows. 
For a permutation $\sigma \in \textrm{S}_n$, the action permutes the tensor factors. 
For a unitary matrix $U \in \textrm{U}(m)$, the action is the $N$-fold tensor product $U^{\otimes n}$. 
It is not hard to see that these two actions commute.
We can now describe Schur-Weyl duality as the following theorem.
\begin{theorem}[Schur-Weyl duality~\cite{rowe2012}]
The Hilbert space of $n$ $m$-dimensional qudits decomposes into irreducible subspaces
\begin{equation}
(\mathbb{C}^m)^{\otimes n} \simeq \bigoplus_{\lambda\vdash n} \mathbb{C}^{\{\lambda\}} \otimes \mathbb{C}^{(\lambda)} ,
\end{equation}
where $\mathbb{C}^{\{\lambda\}}$ carries irrep $\{\lambda\}$ of $\textrm{U}(m)$ and $\mathbb{C}^{(\lambda)}$ carries irrep $(\lambda)$ of $\textrm{S}_n$, and $\simeq$ indicates a change of basis is involved.
The dimension of irrep $(\lambda)$ can be viewed as the multiplicity of irrep $\{\lambda\}$, and vice versa.
\end{theorem}

There is an efficient quantum circuit that implements the Schur-Weyl decomposition.
Given a state $|\Psi\rangle \in (\mathbb{C}^m)^{\otimes n}$ in the computational basis, this circuit, which we label $W$, performs the transformation
\begin{equation}
W|\Psi\rangle
 = \sum_{\lambda \vdash n} \, \sum_{q_{\lambda}} \, \sum_{p_\lambda}C^\lambda_{q_\lambda,p_\lambda}|\lambda\rangle|q_\lambda\rangle|p_\lambda\rangle , 
\end{equation}
where $\lambda$ indexes the irrep, $q_\lambda$ and $p_\lambda$ index bases of irreps $\{\lambda\}$ and $(\lambda)$ respectively, and $C^\lambda_{q_\lambda,p_\lambda}$ is a generalised Clebsch-Gordan coefficient.
For example, the unitary action of U$(m)$ in this basis is
\begin{equation}
U : |\lambda\rangle |q_\lambda\rangle |p\rangle \rightarrow |\lambda\rangle |U\cdot q_\lambda\rangle |p\rangle := |\lambda\rangle \left( \sum_{q_\lambda'} U^\lambda_{q_\lambda, q_\lambda'} |q_\lambda'\rangle \right) |p\rangle ,
\end{equation}
where $U^\lambda$ is the irreducible unitary matrix corresponding to $U \in \mathrm{U}(m)$.
It was proven by Bacon, Chuang and Harrow that this circuit runs in polynomial time in terms of $n$, $m$ and $\log(\delta^{-1})$, where $\delta$ is an accuracy parameter~\cite{bacon2007}.

\section{A quantum circuit for ideal bosonic sampling}
\label{sec:qc-bs}

Here we describe a quantum circuit for bosonic sampling when the bosons are perfectly indistinguishable (and free from other errors such as loss, which we'll discuss later on). 
This circuit samples with accuracy $\delta + \epsilon$ and runs in polynomial time and space in terms of $m$, $n$, $\log(\delta^{-1})$ \& $\log(\epsilon^{-1})$. 
Here and throughout the paper, $\delta$ describes the precision with which we are able to approximate the Schur transform via the Bacon-Chuang-Harrow circuit, and $\epsilon$ the accuracy to which we can approximate the unitary matrix $U$ via the Solovay-Kitaev theorem \cite{nielsen2010,dawson2006}.
Note that although the Solovay-Kitaev construction can involve exponential resources in terms of $m$, this can be avoided by first performing a Hurwitz (or Reck) decomposition into smaller unitaries~\cite{hurwitz1897, reck1994, barenco1995}.

The goal of the circuit is to sample from the totally symmetric subspace of $(\mathbb{C}^m)^{\otimes n}$, where the interferometer $U \in \textrm{U}(m)$ acts as the totally symmetric irrep of the unitary group given by $\{ \lambda = (n)\}$.
In order to construct symmetrised states given the input occupation $S$, we use the (inverse) Schur transform. 
The Schur circuit $W$ specifies irreps of U$(m)$ in the Gelfand-Zeitlin (GZ) basis, so we need a way to map between these states and occupations. 
We can do this via the pattern weight $\nu = (\nu_1,\cdots,\nu_d)$, which can be related to a GZ pattern for any irrep~\cite{alex2011}. 
For the fully symmetric irrep, the pattern weight is unique for each GZ state and there is a particularly simple mapping from occupations to symmetric GZ states in this case, namely $\nu_i = S_i$~\cite{rowe1999}; this has also been referred to as a quantum analog of a classical ``type''~\cite{harrow2005}.
Thus, we have an efficient way to identify an input occupation $S$ with a GZ basis state $q_{(n)}$.

We can now see how a circuit for indistinguishable boson sampling would work. 
Given an input occupation $S$, we prepare the corresponding state $|q_{(n)}\rangle$ of the $q$-register. 
To use the inverse Schur transform, we append to this input state a quantum register for the irrep $|(n)\rangle$, and another for the symmetric group index $|p_{(n)}\rangle$. 
Note that there is only one possible state for the $p_{(n)}$ register, because the fully symmetric irrep of the symmetric group is one dimensional; thus $p_{(n)}=1$ always.
The inverse Schur transformation $W^\dag$ takes this state to a symmetric state of $n$ qudits in $(\mathbb{C}^m)^{\otimes n}$.
 In this tensor space, we now need only apply the interferometer matrix $U$ to each qudit in parallel as the circuit $U^{\otimes n}$.
This can be done with accuracy $\epsilon$ in $O(\log^c(1/\epsilon))$ time via the Solovay-Kitaev theorem~\cite{nielsen2010, dawson2006}. 
Finally, we apply the Schur transform again and measure the $q$-register to get a sample $q_{(n)}'$, from which we can easily compute the pattern weight/type to get an output occupation $S'$.

A complete version of the quantum circuit for Boson Sampling is given in Algorithm \ref{alg:bs}, as well as a circuit description in Figure \ref{fig:bs}.
\begin{algorithm}
\SetKwInOut{Input}{input}\SetKwInOut{Output}{output}
\Input{A matrix $U \in \text{U}(m)$ and \\
an $n$-boson $m$-mode occupation $S$.}
\Output{An $n$-boson $m$-mode occupation $S'$.}
\BlankLine
Map $S$ to $q$-register basis index $q_{(n)}$\;
Prepare input state $|\lambda=(n)\rangle|q_{(n)}\rangle|p_{(n)}=1\rangle$\;
Apply $W^\dagger$, producing a state $|\Psi\rangle \in (\mathbb{C}^m)^{\otimes n}$\;
Synthesize $U$ via Solovay-Kitaev\;
Execute $U$ on each qudit in parallel, implementing $U^{\otimes n}$\;
Apply $W$, producing state $|(n)\rangle|U\cdot q_{(n)}\rangle|1\rangle$\; 
Measure the $q$-register to obtain a sample $q_{(n)}'$\;
Map $q_{(n)}'$ to an occupation $S'$\;
\Return{$S'$}
\caption{
A quantum circuit for sampling from the same distribution as that produced by indistinguishable bosons in a linear interferometer.}
\label{alg:bs}
\end{algorithm}

\begin{figure}
\includegraphics[width=\linewidth]{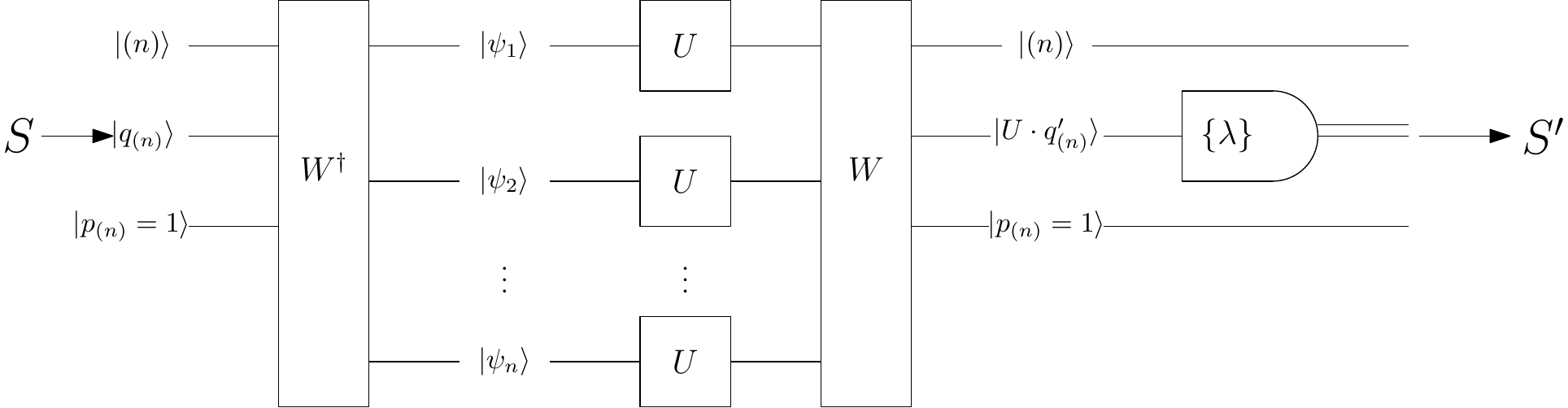}
\caption{
A quantum circuit for Algorithm \ref{alg:bs}. 
Note that the measurement of the $q$-register returns a string that we associate to a GZ basis state.}
\label{fig:bs}
\end{figure}

We can demonstrate correctness by showing that this distribution does indeed match the one we have for sampling from indistinguishable bosons. 
We start with the input occupation $S$. 
After mapping this to a unitary irrep state $|q_{(n)}\rangle$ and applying $W^\dagger$, we end up with the symmetrized state
\begin{equation}\label{eq:IdealState}
W^\dagger|q_{(n)}\rangle
 = \frac{1}{\sqrt{n! \prod_{i=1}^m S_i!}}\sum_{\sigma \in  \textrm{S}_n}\sigma|s\rangle,
\end{equation}
where $|s\rangle$ can be chosen to be any computational basis state with occupation $S$, that is, with $S_i$ of the qudits in state $i$.
Arguing similarly for the output $S'$, we see that the probability of the algorithm outputting $q_{(n)}'$ given inputs $U$ and $S$ is
\begin{align}
\textrm{Pr}[q_{(n)}']
&= \vert \langle q_{(n)}'\vert W U^{\otimes n} W^\dag |q_{(n)}\rangle \vert^2 \\
&= \left|\frac{1}{n!\sqrt{\prod_{i=1}^m S'_i! S_i!}}\sum_{\sigma, \sigma'\in \textrm{S}_n}\langle s'|{\sigma}'^\dag{U}^{\otimes n}{\sigma}|s\rangle\right|^2\\
&= \frac{1}{\prod_{i=1}^m S'_i! S_i!}\left|\frac{1}{n!}\sum_{\sigma, \sigma'\in \textrm{S}_n}\bigotimes_{k=1}^n \langle s'_{\sigma'^{-1}(k)}|{U}|s_{\sigma^{-1}(k)}\rangle\right|^2\\
&= \frac{1}{\prod_{i=1}^m S'_i! S_i!}\left|\frac{1}{n!}\sum_{\sigma, \sigma'\in \textrm{S}_n}\prod_{k=1}^n U_{s'_{k},s_{\sigma^{-1}(\sigma'(k))}}\right|^2\\
&= \frac{1}{\prod_{i=1}^m S'_i! S_i!}\left|\sum_{\tau \in \textrm{S}_n}\prod_{k=1}^n U_{s'_{k},s_{\tau(k)}}\right|^2\\
&= \frac{|\per(U_{S',S})|^2}{\prod_{i=1}^m S'_i! S_i!}.
\end{align}
Thus the output probability distribution matches the one given in Eq.~(\ref{eqn:bs-distribution}).
We also see that Schur-Weyl duality implies
\begin{equation}
U^{(n)}_{q_{(n)}, q_{(n)}'} = \frac{\per(U_{S',S})}{\sqrt{\prod_{i=1}^m S'_i! S_i!}}.
\label{eqn:symmetric_action}
\end{equation}
That is, the totally symmetric representation of the unitary group can be constructed from permanents of $U_{S',S}$ matrices~\cite{bhatia1997}.

As for the complexity of this circuit, the mapping from bosons to $q_{(n)}$ states and back can be done in polynomial time and space in terms of $n$~\cite{rowe1999}, Schur-Weyl duality takes polynomial time and space in terms of $d$, $n$ \& $\log(\delta^{-1})$ and the Solovay-Kitaev theorem allows $U$ to be implemented in polynomial time and space. 
From this and the earlier points discussed in this section, we find that Theorem \ref{thm:bs} holds.
\begin{theorem}
Algorithm \ref{alg:bs} performs ideal bosonic sampling with approximation $\delta+\epsilon$ in polynomial time and space in terms of $m$, $n$, $\log(\delta^{-1})$ and $\log(\epsilon^{-1})$.
\label{thm:bs}
\end{theorem}

We observe that this circuit could be simplified in several ways.
Firstly, the entire Schur transform is not required because in the ideal case the problem is confined to the symmetric irrep only.
As we will see, this is not the case in the non-ideal (distinguishable) case.

Another simplification that we'll use in the next section is the following observation.
If rather than applying step 5 onwards in Algorithm \ref{alg:bs} we simply measure the registers in the computational basis and return the result rewritten as an occupation, we also end up with a distribution given by the permanents. 
The probability of measuring a particular computational basis state $|s'\rangle \in (\mathbb{C}^m)^{\otimes n}$ is
\begin{align}
\textrm{Pr}[s'] 
&= \frac{1}{n!\prod_{i=1}^m (S_i!)}\left|\langle s'|U^{\otimes n}\sum_{\sigma \in \textrm{S}_n}\sigma|s\rangle\right|^2\\
&= \frac{1}{n!\prod_{i=1}^m (S_i!)}\left|\sum_{\sigma \in \textrm{S}_N}\bigotimes_{k = 1}^n\langle s_k'|U|s_{\sigma^{-1}(k)}\rangle\right|^2\\
&= \frac{\left|\sum_{\sigma \in \textrm{S}_n}\prod_{k = 1}^n U_{s_k', s_{\sigma^{-1}(k)}}\right|^2}{n!\prod_{i=1}^m (S_i!)} \\
&= \frac{|\per(U_{S',S})|^2}{n!\prod_{j=1}^m (S_i!)}.
\end{align}
%This becomes the Boson Sampling distribution by mapping back to occupations. 
The probability of measuring an occupation $S'$ is equal to summing over the probabilities of all states $|s'\rangle$ of type $S'$, of which there are $n!/\prod_{i=1}^m (S_i'!)$.
We will consider both versions of this circuit in subsequent sections on sampling from distinguishable bosonic distributions.

\section{A quantum circuit for arbitrarily distinguishable bosonic sampling}
\label{sec:qc-noisy}

We now turn to the question of sampling from a distribution of partially distinguishable bosons, (again with no loss).
Distinguishability is modelled as correlation between the modes of the bosons' `System'  degrees of freedom, and new modes corresponding to `Label' degrees of freedom.
In order to accommodate the possibility of all $n$ bosons being completely distinguishable, the number of Label modes must be $n$ so that each boson can be correlated to a unique Label.
Thus there are now a total of $mn$ modes in the problem.
Physically we can think of the System degree of freedom as the spatial modes available to the bosons, and the Label as, say, temporal modes -- however the model is completely general. 

On the aggregate Hilbert space we have the same setup as the ideal case, but now by tracing out the Label register we see that distinguishability can lead to decoherence of the System qudits.
We assume that an interferometer implementing a $m \times m$ unitary matrix acts only on the $m$ System modes, while the Label remains unchanged. 
In this model, as well as receiving a unitary matrix $U$ as input, we also receive an $m \times n$ occupation $T$ which describes how many bosons are in System mode $i$ and Label mode $j$. 
This can be described in terms of creation operators as
\begin{equation}
|T\rangle = \prod_{i=1}^m \prod_{j = 1}^n \frac{(a_{i,j}^\dagger)^{T_{i,j}}}{\sqrt{T_{i,j}!}}|0\rangle.
\end{equation}
Since the Labels are assumed to be unaffected by the interferometer, the creation operators evolve as $a_{i, j}^\dagger \rightarrow \sum_{k = 1}^m U_{i,k}a_{k, j}^\dagger$.

Our technique for handling distinguishable bosons is similar to the ideal case where we consider the fully symmetric irrep $\{(n)\}$ of the Unitary group. 
However, the introduction of label degrees of freedom means that we can no longer map onto the $\{(n)\}$ irrep of $\textrm{U}(m)$. 
%Instead, we use a technique similar to the one employed by Adamson et al.~\cite{adamson2008} and later by Turner~\cite{turner2016}. 
Instead we the $\{(n)\}$ irrep of the aggregrate unitary group $\textrm{U}(mn)$, by mapping a boson in spatial mode $i$ and temporal mode $j$ to pattern weight $(i,j) \equiv (n(i-1) + j)$.

When we apply $W^\dagger$ to the input, we find the same Young symmetrizer as before, but now output a state $|\Psi\rangle \in (\mathbb{C}^m\otimes\mathbb{C}^n)^{\otimes n}$. 
We can think of this as the ideal case but now with each qudit being $mn$ dimensional.
Furthermore each System-Label qudit can be viewed as bipartite, with a $m$-dimensional qudit describing the System degree of freedom and another $n$-dimensional qudit describing the Label. 
We can therefore split the $n$ System-Label qudits into two registers, with the interferometer action and boson detection taking place on only the System register, while the Label register `eavesdrops'.

A complete description of the circuit is given in Algorithm \ref{alg:noisy-bs}, with a circuit diagram given in Figure \ref{fig:noisy-bs}.
\begin{algorithm}
\SetKwInOut{Input}{input}\SetKwInOut{Output}{output}
\Input{A matrix $U \in \textrm{U}(m)$ and \\ an $n$-boson $mn$-mode occupation $T$.}
\Output{An $n$-boson $m$-mode occupation $S'$.}
\BlankLine
Map $T$ to $q$-register basis index $q_{(n)}$ (for U$(mn)$)\;
Prepare input state $|\lambda=(n)\rangle|q_{(n)}\rangle|p_{(n)}=1\rangle$\;
Apply $W^\dagger$, producing a state $|\Psi\rangle \in (\mathbb{C}^{m \times n})^{\otimes n}$\;
Rearrange into two (possibly entangled) quantum registers $|\Psi_{\textrm{Sys}}\rangle = |\psi_{\textrm{Sys}, 1}\rangle\dots|\psi_{\textrm{Sys}, n}\rangle \in (\mathbb{C}^m)^{\otimes n}$ and $|\Psi_{\textrm{Lab}}\rangle = |\psi_{\textrm{Lab}, 1}\rangle\dots|\psi_{\textrm{Lab}, n}\rangle \in (\mathbb{C}^{n})^{\otimes n}$\;
%\tcc{(Note that these registers might be entangled)}
Synthesize $U$ via Solovay-Kitaev\;
Execute $U$ on the $|\Psi_{\textrm{Sys}}\rangle$ qudits in parallel, implementing  $U^{\otimes n} \otimes \Id_\mathrm{Lab}$\;
Measure the System in the computational basis to obtain a sample $s'$\;
Map $s'$ to an occupation $S'$ ($S'_i = \#$ of qudits in state $1\leq i\leq m$)\;
\Return $S'$
\caption{A quantum circuit for sampling from (essentially) the same distribution at that produced by distinguishable bosons in a linear interferometer.
In order to sample from exactly the same distribution, instead of step 7 one could transform back to the Schur basis by applying $W$ on the System and sample the $q$-register, or one could perform some post-processing as discussed at the end of the previous section.}
\label{alg:noisy-bs}
\end{algorithm}

\begin{figure}
\includegraphics[width=\linewidth]{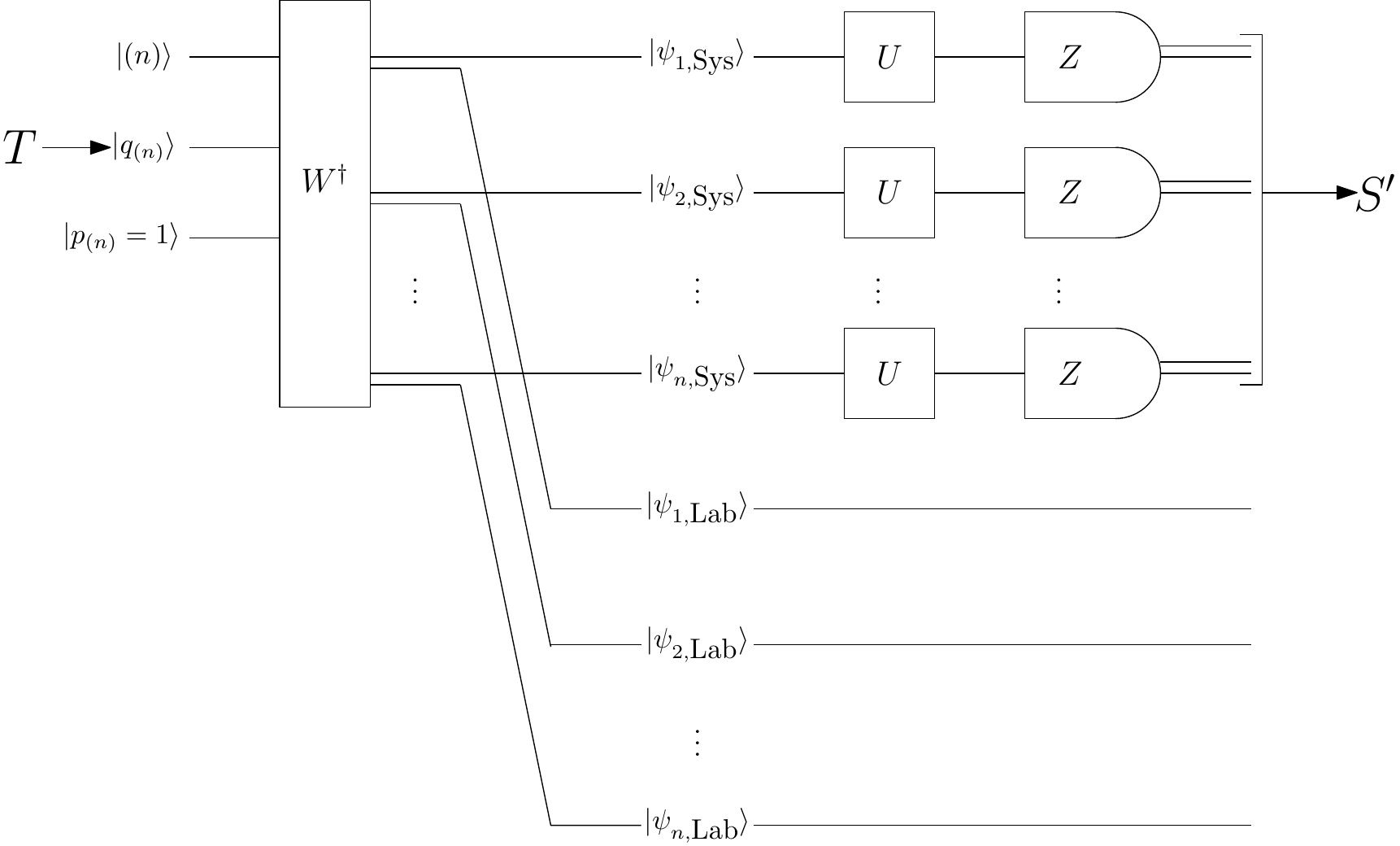}
\caption{The quantum circuit described in Algorithm \ref{alg:noisy-bs}. 
For simplicity we forego the second Schur transform as discussed in the previous section and measure in the computational ($Z$) basis. 
Note that we only sample from the System qudits, effectively tracing out the Label qudits.}
\label{fig:noisy-bs}
\end{figure}

To see that this distribution matches that of partially distinguishable bosons, we will compare with the results of Tichy~\cite{tichy2015}.
There, each boson is assigned an ``internal'' state $|\Phi_i\rangle$, where $i=1, \cdots, m$.
Thus, in our terminology, every boson in System mode $i$ has a Label state given by $\Phi_i$.
The probability distribution given there for sampling from partially distinguishable bosons is (in our notation)
\begin{equation}\label{eq:TichyDist}
\textrm{Pr}[S'] =  \frac{1}{\prod_{i=1}^m S_i! S'_i!} \sum_{\tau, \tau' \in \textrm{S}_n} \prod_{k=1}^n U_{s'_k,s_{\tau(k)}} U^\ast_{s'_k,s_{{\tau'}(k)}} \mathcal{S}_{{\tau'(k)},{\tau(k)}},
\end{equation}
where $\mathcal{S}_{k,l} = \langle \Phi_{s_k}|\Phi_{s_l}\rangle$ is a Hermitian positive-definite $n \times n$ distinguishability matrix, $s_k$ gives the mode occupied by particle $k$, and we've included a factor to account for the possibility of multiple bosons in the same output mode.
In order to connect our model with this model of distinguishability, one simply needs to take superpositions of $mn$-mode input occupations $T$ in such a way as to realise the Label states $\Phi_i$.
This is always possible since the space of internal states, span$\{\Phi_i\}_{i=1}^m$, can always be embedded in the Label space $(\mathbb{C}^n)^{\otimes n}$.
For example, consider two bosons in two System modes where one boson is in System mode 1 and has internal state $|1\rangle$ (corresponding to Label mode 1), and the other boson is in System mode 2 and has internal state $|\Phi_2\rangle=\alpha|1\rangle+\beta|2\rangle$.
This is represented as the following superposition of System-Label occupations (recall rows of $T$ correspond to the System and columns to the Label):
\begin{align}
\alpha \Ket{\begin{matrix} 1&0\\1&0 \end{matrix}} + \beta \Ket{\begin{matrix} 1&0\\0&1 \end{matrix}} .
\end{align}
Given a distinguishability matrix $\mathcal{S}_{k,l}$, in this way we can prepare a corresponding superposition of occupations $T$ at step 2 of the algorithm.
(We can in fact consider more general partially distinguishable situations of bosons with different Label states in the same System mode.)

After step 3 of the algorithm, the state can be written as
\begin{align}
|\Psi\rangle_{\textrm{Sys,Lab}}
 &= \frac{1}{\sqrt{n!\prod_{i=1}^m S_i!}}\sum_{\sigma\in\textrm{S}_n}\sigma|s\rangle\sigma\bigotimes_{k=1}^n|\Phi_{s_k}\rangle\\
 &= \frac{1}{\sqrt{n!\prod_{i=1}^m S_i!}}\sum_{\sigma\in\textrm{S}_n}\sigma|s\rangle\bigotimes_{k=1}^n|\Phi_{s_{\sigma^{-1}(k)}}\rangle\\
\end{align}
where $|s\rangle=\bigotimes_{k=1}^n|s_k\rangle$ and $\bigotimes_{k=1}^n|\Phi_{s_k}\rangle$ are quantum registers describing the System and Label. 
Tracing out the Label register yields the reduced density matrix
\begin{align}
\rho_{\textrm{Sys}}
 &= \textrm{Tr}_{\textrm{Lab}}[|\Psi\rangle_{\textrm{Sys,Lab}}\langle\Psi|]\\
 &= \frac{1}{n!\prod_{i=1}^m S_i!}\sum_{\sigma,\sigma'\in\textrm{S}_n}\sigma|s\rangle\langle s|\sigma'^\dagger \prod_{k=1}^n \langle\Phi_{s_{\sigma'^{-1}(k)}}|\Phi_{s_{\sigma^{-1}(k)}}\rangle\\
 &=\frac{1}{n!\prod_{i=1}^m S_i!}\sum_{\sigma,\sigma'\in\textrm{S}_n}\sigma|s\rangle\langle s|\sigma'^\dagger \prod_{k=1}^n \mathcal{S}_{\sigma'^{-1}(k),\sigma^{-1}(k)}.
\end{align}
When we apply the action of the interferometer $U$ on the $m$-dimensional System qudits, the probability of measuring a state $|s'\rangle$ after step 7 is
\begin{align}
\textrm{Pr}[|s'\rangle]
 &= \textrm{Tr}[|s'\rangle\langle s'|U^{\otimes n}\rho_{\textrm{Sys}}(U^\dagger)^{\otimes n}]\\
 &= \frac{\sum_{\sigma, \sigma' \in \textrm{S}_n}\langle s'|U^{\otimes n}\sigma|s\rangle\langle s|\sigma'^\dagger(U^\dagger)^{\otimes n}|s'\rangle \prod_{k=1}^n\mathcal{S}_{\sigma'^{-1}(k),\sigma^{-1}(k)}}{n!\prod_{i=1}^m S_i!}\\
 &= \frac{\sum_{\sigma, \sigma' \in \textrm{S}_n} \prod_{k=1}^n \langle s'_k|U|s_{\sigma^{-1}(k)}\rangle\langle s_{\sigma'^{-1}(k)}|U^\dagger|s'_k\rangle\mathcal{S}_{\sigma'^{-1}(k),\sigma^{-1}(k)}}{n!\prod_{i=1}^m S_i!}\\
 &= \frac{\sum_{\tau, \tau' \in \textrm{S}_n} \prod_{k=1}^n U_{s'_k,s_{\tau(k)}} U^*_{s'_k,s_{\tau'(k)}} \mathcal{S}_{\tau'(k),\tau(k)}}{n!\prod_{i=1}^m S_i!}.
\end{align}
Up to a factor, this is the desired probability distribution of Eq.(\ref{eq:TichyDist}).
As discussed at the end of the last section, this factor could be handled either by applying a second Schur transform on the System and sampling the $q$-register, or by classical post processing.

Counting resources goes much the same as it did in the ideal case, though now we have $mn$-dimensional qudits that are made up of pairs of $m$- and $n$-dimensional qudits.
Separating these System and Label registers in step 3 can be done with polynomial resources, as can the unitary transformation on the input $q$-register that prepares the input state of arbitrary distinguishability.
From this and the points above, we find that Theorem \ref{thm:noisy-bs} holds.
\begin{theorem}
Algorithm \ref{alg:noisy-bs} samples from the distinguishable bosonic distribution with approximation $\delta+\epsilon$ when the distinguishability of the input bosons is known. 
The circuit runs in polynomial time and space in terms of $m, n, \log(\delta^{-1})$ and $\log(\epsilon^{-1})$.
\label{thm:noisy-bs}
\end{theorem}

\subsection{Complete (in)distinguishability}

The two extreme cases of completely indistinguishability and complete distinguishability are of interest.
For completely indistinguishable bosons, the Label for each is the same (call it $\ket{1}$), and after the (inverse) Schur transform step we have
\begin{equation}
\frac{1}{\sqrt{n!\prod_{i=1}^m S_{i}!}}\sum_{\sigma\in \textrm{S}_n}\sigma|s\rangle|1\rangle^{\otimes n},
\end{equation}
It is clear that the Label register is separable from the System register, and tracing out the Label yields the same situation as the ideal case in Eq.(\ref{eq:IdealState}), as it should.

In the completely distinguishable case, each boson has a different unique Label, correlated to a unique System mode (note $m\geq n$ in this case).
This implies that the occupation $T$ has a single 1 in each of $n$ rows and columns, and zeros elsewhere.
The (inverse) Schur transformed state has System and Label registers $s$ and $l$ that are completely correlated sequences of length $n$ with no repetitions; if we choose to order the bases $123\cdots n$ then we have the state
\begin{align}\label{eq:ComDis}
\frac{1}{\sqrt{n!}} \sum_{\sigma \in \textrm{S}_n}\sigma|s\rangle\sigma|l\rangle
&=\frac{1}{\sqrt{n!}} \sum_{\sigma \in \textrm{S}_n}|\sigma^{-1}(1) \sigma^{-1}(2) \cdots \sigma^{-1}(n)\rangle |\sigma^{-1}(1) \sigma^{-1}(2) \cdots \sigma^{-1}(n)\rangle .
\end{align}
We see that this is maximally entangled on the ``coincident'' subspace of states with a single boson in each mode.
Tracing out the Label yields the reduced System state
\begin{align}
\rho_\mathrm{Sys}
 &= \frac{1}{n!} \sum_{\sigma, \sigma' \in \textrm{S}_n} \sigma|s\rangle\langle s|\sigma'^\dagger \, \langle l|\sigma'^\dagger \sigma|l\rangle \\
 &= \frac{1}{n!} \sum_{\sigma \in \textrm{S}_n} \sigma|s\rangle\langle s|\sigma^\dagger ,
\end{align}
which follows because the Label overlap is only nonzero if $\sigma'^\dag \sigma = \Id \Rightarrow \sigma'=\sigma$ due to the fact that $l$ has no repetitions.
After the action of $U$, the probability of measuring $|s'\rangle \in (\mathbb{C}^m)^{\otimes n}$ is
\begin{align}
\textrm{Pr}[|s'\rangle]
 &= \textrm{Tr}[|s'\rangle\langle s'|U^{\otimes n}\rho U^{\dagger\otimes n}]\\
 &= \frac{1}{n!} \langle s'|U^{\otimes n}\left(\sum_{\sigma \in  \textrm{S}_n}\sigma|s\rangle\langle s|\sigma^\dagger\right)U^{\dagger\otimes n}|s'\rangle\\
% &= \frac{1}{n!} \sum_{\sigma \in  \textrm{S}_n} \prod_{k=1}^n\langle s_k'|U|s_{\sigma^{-1}(k)}\rangle\langle s_{\sigma^{-1}(k)}|U^\dagger|s_k'\rangle\\
 &= \frac{1}{n!} \sum_{\sigma \in  \textrm{S}_n} \prod_{k=1}^n |U_{s'_k,s_{\sigma^{-1}(k)}}|^2\\
 &= \frac{\textrm{per}(|U_{S',S}|^2)}{n!} ,
\end{align}
where we've defined $|U_{S',S}|^2$ as the elementwise square of the absolute value.
We can find the probability of returning occupation $S'$ by summing up the probabilities of all $n!/\prod_{i=1}^{m}(S'_i!)$ states of type $s'$, giving
\begin{equation}\label{eq:DistPer}
\textrm{Pr}[S'] = \frac{\textrm{per}(|U_{S',S}|^2)}{\prod_{i=1}^m S'_i!},
\end{equation}
which agrees with the (classical) probability distribution for sampling with distinguishable bosons~\cite{aaronson2014}.

Note that we could have considered a distinguishable input where each boson has a unique Label, but with multiple occupancy of System modes.
In that case the analysis would show the output distribution to be the same as above up to a factor of $\prod_i S_i$ in the denominator.
%{\blue I think this subsection is pretty clear, the reader should be able to produce the 2 photon calculations themself so probably not necessary to include the appendix.}

\section{Unitary-unitary duality}
\label{sec:uu-duality}

The preceding shows how the Schur transform gives a map between second quantised occupation states and first quantised single particle states via symmetrisation.
The complication added by distinguishability is that each single particle becomes bipartite, with a System and Label degree of freedom.
As shown above, distinguishability arises as correlations between the System and Label registers of the circuit in Fig.~\ref{fig:noisy-bs}.
It turns out that independently transforming the System and Label registers back into the Schur basis can give a Schmidt decomposition for these correlated states (see Fig.~\ref{fig:noisy-bs-rep}).
This can be seen as a consequence of the following duality~\cite{goodman2009, rowe2012}.
\begin{theorem}[Unitary-unitary duality]\label{thm:unitary-unitary}
The totally symmetric irrep of $\mr{U}(md)$ can be decomposed into irreps of $\mr{U}(m) \times \mr{U}(d)$ as
\begin{align}\label{eq:uudecomp}
(\mathbb{C}^{m} \otimes \mathbb{C}^{d})^{(n)} \cong \bigoplus_{\lambda\vdash n} \mathbb{C}^{\{\lambda\}_m} \otimes \mathbb{C}^{\{\lambda\}_d} ,
\end{align}
where $\{\lambda\}_m$ indicates an irrep $\lambda$ of $\mr{U}(m)$, similarly for $\{\lambda\}_d$, and $\lambda$ runs over all partitions of $n$ consistent with both $m$ and $d$.
\end{theorem}
This can be proven by Schur decomposing the System and Label registers, each of which, by Schur-Weyl duality, will have good permutation symmetry quantum numbers.
The question is then which linear combinations of tensor products of such states are totally symmetric, and the answer turns out to be only states of the form (suppressing normalisation)~\cite{hamermesh1962}
\begin{align}\label{eq:symSL}
\ket{\lambda, q_\lambda, {q'}_\lambda}_\mr{SysLab} &=  \sum_{p_\lambda} \ket{\lambda, q_\lambda, p_\lambda}_\mr{Sys} \ket{\lambda, {q'}_\lambda, p_\lambda}_\mr{Lab} .
\end{align}
Thus a basis for the totally symmetric irrep of the System-Label Hilbert space consists of states of this form, leading to the decomposition in Eq.(\ref{eq:uudecomp}) and to a Schmidt decomposition of totally symmetric (second quantised) System-Label states.

Because the Schur transformations on each register are local to the System and Label, entanglement across this bipartition is unchanged.
For example, the completely distinguishable state that from Eq.(\ref{eq:ComDis}) is seen to have Schmidt rank $n!$ in the computational basis, and therefore must be of the form
\begin{equation}
\frac{1}{\sqrt{n!}} \sum_{\lambda \vdash n} \sum_{q_\lambda^\mr{coin}} \sum_{p_\lambda} |\lambda,q_\lambda,p_\lambda\rangle_{\textrm{Sys}}|\lambda,q_\lambda,p_\lambda\rangle_{\textrm{Lab}},
\end{equation}
where the sum over $q_\lambda^\mr{coin}$ is taken over the coincident subspace, that is, the irrep basis states with pattern weight (or type) $(11\cdots 1)$.
That the dimensions of these spaces are the same can be shown combinatorially and follows from the theorem.
This shows that although in the ideal case only the symmetric subspace is in play and therefore the full Schur transform is overkill, for the distinguishable case all irreps $\lambda$ can play a role.

\begin{figure}
\includegraphics[width=\linewidth]{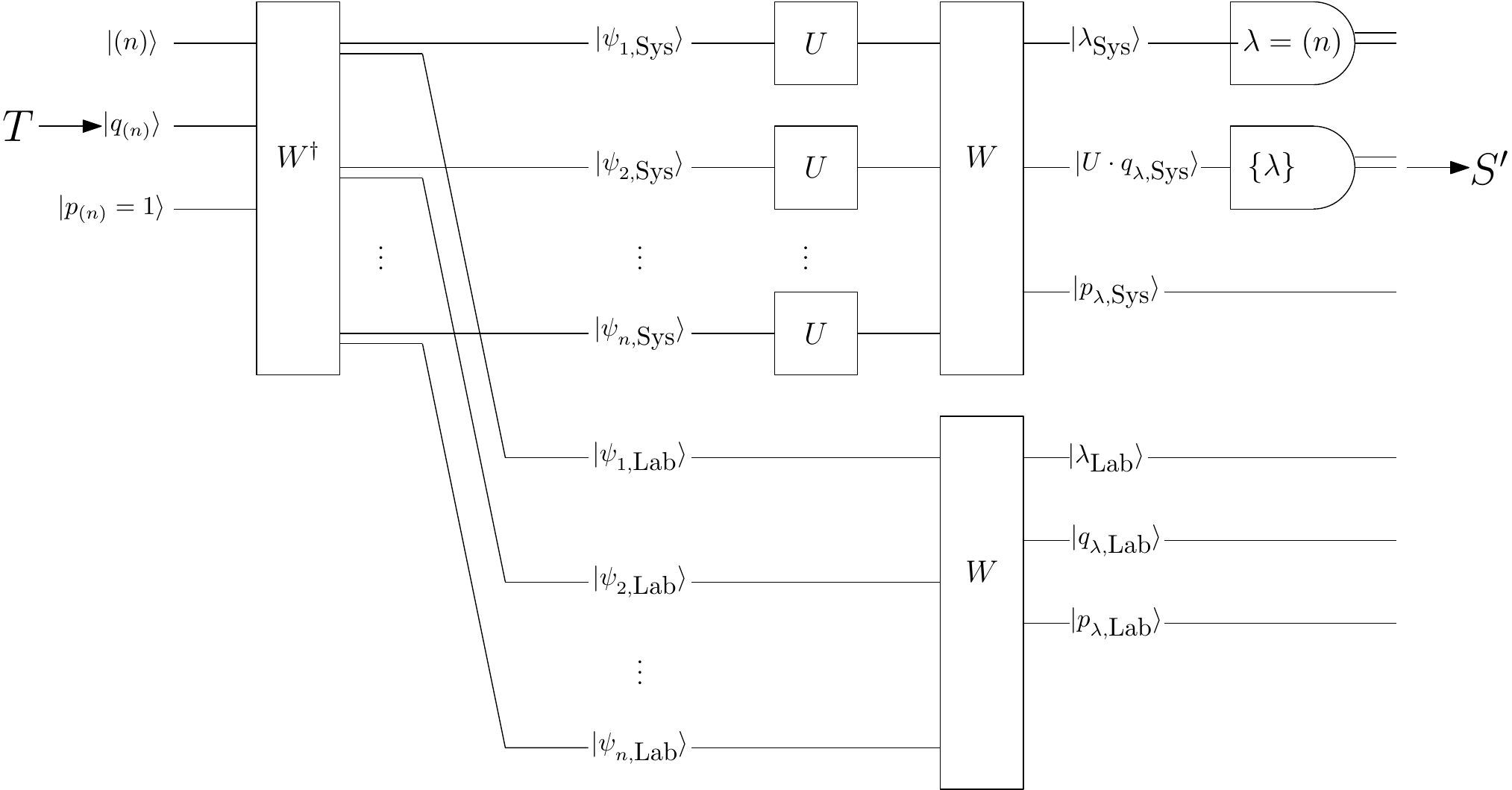}
\caption{Circuit diagram illustrating how postselection can be used to `filter out' distinguishability.
Note that the postselection measurement of the System $\lambda$-register is in the irrep basis, while that of the System $q$-register is in the GZ basis.
The Schur transform on the Label register is not necessary, but illustrates the unitary-unitary duality.}
\label{fig:noisy-bs-rep}
\end{figure}

\subsection{Postselection of ideal bosonic sampling}
\label{sec:post-bs}

Although unitary-unitary duality can be demonstrated in this model by implementing local Schur transforms before measuring, in both the ideal and distinguishable case circuits considered previously, it was argued that this was not necessary; it was enough to measure in the computational basis after the unitary transformation was implemented and post-process.
An interesting observation is that given a distinguishable input, by performing the second Schur transform on the System it becomes possible to use postselection to sample from the indistinguishable distribution.
Of course, this comes at the cost of throwing away a lot of bad samples.

To achieve this postselective filtering, we need to ensure that we only sample the System from the fully symmetric irrep of $\textrm{U}(m)$. 
This is done by measuring the irrep register $|\lambda\rangle_{\textrm{Sys}}$ and waiting for the outcome $\lambda_{\textrm{Sys}}=(n)$. 
After postselection, the amplitudes of the System $q$-register $|U\cdot q_{(n)}\rangle$ are given by Eq.~\ref{eqn:symmetric_action}, which give the same probability distribution as sampling indistinguishable bosons. 
Following the arguments above, the circuit remains efficient since the added Schur transform can be implemented efficiently. 

More generally, such a postselected quantum circuit could sample from any irrep $\lambda$ of $\textrm{U}(m)$. 
All we need to achieve this is to ensure that the input state has support in the irrep we wish to sample from, and postselect on being in that irrep.
A dimension counting argument shows that the completely distinguishable input discussed above has support in all irreps~\cite{stanisic2018}, and so could be used for this purpose.
%We will discuss the complexity of such sampling in Sec.~\ref{sec:conclusion}.

\section{Sampling with loss}
\label{sec:loss}

Another serious practical difficulty with linear optics is the loss of photons through unwanted scattering processes. 
In this section, we discuss how the loss model developed by Aaronson and Brod~\cite{aaronson2015} can be simulated. 
In their model, $n+k$ bosons are generated as occupation $S^0$, $k$ of which are lost before they reach the interferometer. 
As we don't know which bosons were lost, Aaronson and Brod take the average over the set of all $n$-boson occupations which are consistent with $S^0$, denoted $\bar{\Lambda}_{m,S^0,n}$.
The result can be shown by tracing out any choice of $k$ qudits in the ideal case, as shown in Fig.~\ref{fig:lost-bs}. 

\begin{figure}
\includegraphics[width=\linewidth]{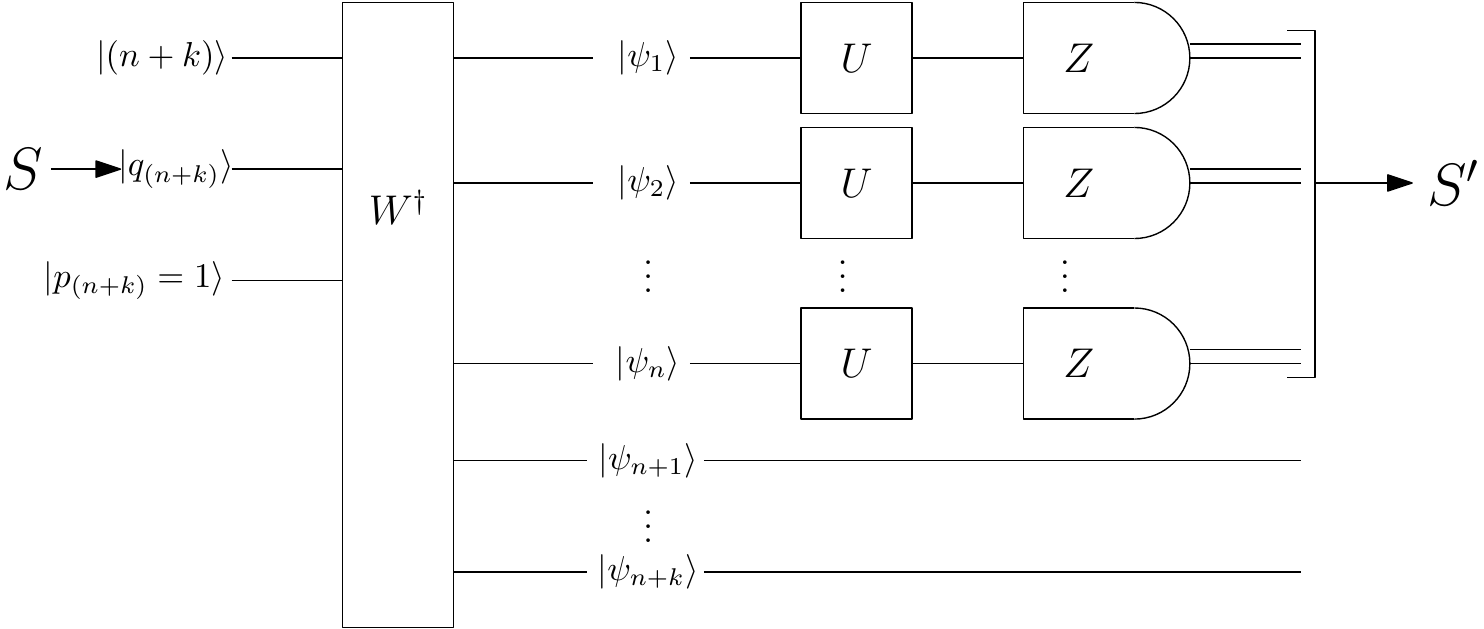}
\caption{Circuit for sampling when $k$ bosons are lost. 
Here, we ignore $k$ qudits of the System register, tracing them out with the Label register. 
As with Fig.~\ref{fig:noisy-bs}, measurements are in the computational basis.}
\label{fig:lost-bs}
\end{figure}

\begin{theorem}
Let $|\psi\rangle$ be the state after step 2 of Algorithm \ref{alg:bs} with $(n+k)$-boson input state $S^0$, and $\bar{\Lambda}_{m,S^0,n}$ be the set of all $n$-boson occupations which are consistent with $S^0$.
If $k$ qudits of $|\psi\rangle$ are traced out before continuing with the algorithm, the final probability distribution of output occupations $S'$, denoted $\mathcal{D}_U$, is
\begin{equation}
\mr{Pr}_{\mathcal{D}_U}[S'] = \frac{1}{\binom{n+k}{k}}\sum_{S \in \bar{\Lambda}_{m,S^0,n}} |\mr{per}(U_{S',S})|^2 \prod_{i=1}^m \frac{\binom{S_i^0}{S_i}}{S'_i! S_i!} .
\end{equation}
\end{theorem}

\begin{proof}
The state $|\psi\rangle$ can be written as the density matrix
\begin{equation}
\rho = |\psi\rangle\langle\psi| = \frac{1}{(n+k)!\prod_{i=1}^m S^0_i!} \sum_{\sigma\in\textrm{S}_{n+k}}\sigma|s^0\rangle\sum_{\sigma'\in\textrm{S}_{n+k}}\langle s^0|\sigma'^\dagger,
\end{equation}
where $|s^0\rangle$ is any state consistent with the input state occupation $S^0$.
This state is symmetric, so the choice of which qudits to trace out is moot.
Choosing the last $k$ qudits, the reduced density matrix for the remaining $n$ particles is
\begin{align}
\rho_n
 &= \frac{1}{(n+k)!\prod_{i=1}^m S^0_i!}\sum_{\sigma,\sigma'\in\textrm{S}_{n+k}}\bigotimes_{l=1}^n|s^0_{\sigma^{-1}(l)}\rangle\langle s^0_{\sigma'^{-1}(l)}|\bigotimes_{l=n+1}^{n+k}\langle s^0_{\sigma'^{-1}(l)}|s^0_{\sigma^{-1}(l)}\rangle\label{eqn:losstrace}\\
 &=\frac{1}{(n+k)!\prod_{i=1}^m S^0_i!}\sum_{\substack{S \subseteq \{s^0_1,\dots,s^0_{n+k}\}\\\bar{S}=\{s^0_1,\dots,s^0_{n+k}\}\setminus{S}\\|S|=n}}\sum_{\sigma \in \textrm{S}_{n}}\sigma|s\rangle\sum_{\sigma' \in \textrm{S}_{n}}\langle s|\sigma'^\dagger\sum_{\tau,\tau'\in\textrm{S}_k}\langle \bar{s}|\tau'^\dagger\tau|\bar{s}\rangle\label{eqn:firstloss}\\
 &=\frac{1}{(n+k)!\prod_{i=1}^m S^0_i!}\sum_{\substack{S \in \Lambda_{m,S^0,n}\\\bar{S}=\{s^0_1,\dots,s^0_{n+k}\}\setminus{S}}}\prod_{i=1}^m\left(\frac{S^0_i!}{S_i!(S_i^0-S_i)!}\right)^2\sum_{\sigma \in \textrm{S}_{n}}\sigma|s\rangle\sum_{\sigma' \in \textrm{S}_{n}}\langle s|\sigma'^\dagger\sum_{\tau,\tau'\in\textrm{S}_k}\langle \bar{s}|\tau'^\dagger\tau|\bar{s}\rangle\label{eqn:secondloss}\\
 &=\frac{k!}{(n+k)!}\sum_{S \in \Lambda_{m,S^0,n}}\frac{\prod_{i=1}^m S^0_i!}{(\prod_{j=1}^m S_j!)^2\prod_{i=1}^m (S^0_i-S_i)!}\sum_{\sigma \in \textrm{S}_{n}}\sigma|s\rangle\sum_{\sigma' \in \textrm{S}_{n}}\langle s|\sigma'^\dagger\label{eqn:thirdloss}\\
 &=\frac{1}{\binom{n+k}{n}}\sum_{S \in \Lambda_{m,S^0,n}}\frac{\prod_{i=1}^m \binom{S_i^0}{S_i}}{n!\prod_{j=1}^m S_j!}\sum_{\sigma \in \textrm{S}_{n}}\sigma|s\rangle\sum_{\sigma' \in \textrm{S}_{n}}\langle s|\sigma'^\dagger , \label{eq:lastloss}
\end{align}
where now $|s\rangle$ (resp.\ $|\bar{s}\rangle$) is any state consistent with the occupation $S$ (resp.\ $\bar{S}$). 
In this calculation we first break the qudits into multisets of $S$ and $\bar{S}$ with respective sizes $n$ and $k$, and permute each multiset individually, which is done in Eq.(\ref{eqn:firstloss}). 
Note that $S \subseteq \{s^0_{1},\dots,s^0_{n+k}\}$ such that $|S|=n$ implies that $S\in\Lambda_{m,S^0,n}$ defined above, so we can sum over $\Lambda_{m,S^0,n}$. 
However, doing so will ignore duplicates of $S$ we had when considering multisets included in $\{s^0_1,\dots,s^0_{n+k}\}$, which need to be acconted for. 
The total number of duplicate terms is the number of permutations $\sigma,\sigma'\in\textrm{S}_{n+k}$ for which $|s^0\rangle$ is invariant, of which there are $(\prod_{i=1}^m S^0_i!)^2$. 
The permutations $\sigma,\sigma'\in\textrm{S}_n$ and $\tau,\tau'\in\textrm{S}_k$ mean that $(\prod_{i=1}^mS_i!(S_i^0-S_i)!)^2$ duplicates are already accounted for. 
Putting these two points together, we get the factor seen in Eq.(\ref{eqn:secondloss}). 
Finally in Eq.(\ref{eqn:thirdloss}), we take the inner product, noting that $\sum_{\tau,\tau'\in\textrm{S}_k}\langle \bar{s}|\tau'^\dagger\tau|\bar{s}\rangle = k!\prod_{i=1}^m(S_i^0-S_i)!$.

Applying the unitary transformation $U$ and measuring in the computational basis, we find that the calculation of the probability of measuring a state $|s'\rangle \in (\mathbb{C}^m)^{\otimes n}$ goes through much as in the previous sections.
Applying the same methods as before, we have
\begin{align}
\textrm{Pr}[|s'\rangle] &= \textrm{Tr}[|s'\rangle\langle s'|(U)^{\otimes n}\rho_n (U^\dagger)^{\otimes n}]\\
%&= \frac{1}{\binom{n+k}{n}}\sum_{S \in \bar{\Lambda}_{d,n+k,n}}\frac{\prod_{i=1}^d\binom{S_i^0}{S_i}}{n!\prod_{j=1}^dS_j!}\langle s'|U^{\otimes n}\left(\sum_{\sigma \in \textrm{S}_n}\sigma|s\rangle\right)\left(\sum_{\sigma' \in \textrm{S}_n}\langle s|\sigma'^\dagger\right)(U^\dagger)^{\otimes n}|s'\rangle\\
&= \frac{1}{\binom{n+k}{n}}\sum_{S \in \bar{\Lambda}_{m,n+k,n}} \frac{\prod_{i=1}^m \binom{S_i^0}{S_i}}{n!\prod_{j=1}^m S_j!} \left|\langle s'|U^{\otimes n}\left(\sum_{\sigma \in \textrm{S}_n}\sigma|s\rangle\right)\right|^2\\
%&= \frac{1}{\binom{n+k}{n}}\sum_{S \in \bar{\Lambda}_{d,n+k,n}}\frac{\prod_{i=1}^d\binom{S_i^0}{S_i}}{n!\prod_{j=1}^dS_j!}\left|\sum_{\sigma \in \textrm{S}_n}\bigotimes_{l=1}^n\langle s'_l|U|s_{\sigma^{-1}(s_l)}\rangle\right|^2\\
%&= \frac{1}{\binom{n+k}{n}}\sum_{S \in \bar{\Lambda}_{d,n+k,n}}\frac{\prod_{i=1}^d\binom{S_i^0}{S_i}\left|\sum_{\sigma \in \textrm{S}_n}\prod_{l=1}^nU_{s'_l,s_{\sigma^{-1}(l)}}\right|^2}{n!\prod_{j=1}^dS_j!}\\
&= \frac{1}{\binom{n+k}{n}}\sum_{S \in \bar{\Lambda}_{m,n+k,n}} |\mr{per}(U_{S',S})|^2 \prod_{i=1}^m \frac{\binom{S_i^0}{S_i}}{n!S_i!} .
\end{align}
To find the probability of sampling occupation $S'$, we add together the probabilities for all computational basis states $|s'\rangle$ that map to $S'$, of which there are $n!/\prod_{i=1}^m S'_i!$. 
This gives us the desired probability distribution.
\end{proof}

Combining loss with distinguishability can be simulated by splitting the remaining $n$ qudits in the state given by Eq.(\ref{eq:lastloss}) into System and Label registers, and tracing out the Label.
This would result in similar averages over the lossless cases described in the sections above.

\section{Distinguishability and simulateability}
\label{sec:mixture}

Our model of distinguishability as correlations with the Label register gives an explicit decoherence model for the computation on the System register.
A natural question is to ask at what point this decoherence renders the quantum computation classically simulateable.
There is a large amount of literature surrounding classical simulation of mixed state quantum computing~\cite{fujii2014, morimae2014classical, morimae2014hardness}, and the role of entanglement~\cite{harrow2003, virmani2005, buhrman2006, vidal2003, animesh2007}. 
We've already seen that when the input is completely distinguishable, the output distribution is given by the permanents of positive matrices, Eq.(\ref{eq:DistPer}), which can be approximated in polynomial time~\cite{anari2017}. 
This efficient permanent approximation method can be used with Clifford and Clifford's algorithm~\cite{clifford2017} to produce a polynomial runtime for approximate sampling. 
Another method for efficiently simulating fully distinguishable photons is to simulate each photon going through the interferometer individually~\cite{aaronson2014,neville2017}. 
%\FIXME{I think~\cite{anari2017} this is the citation you are referring to but I don't know for sure. Also I would be wary of this argument; while the permanents might be easy to compute, there are at least like $\binom{m}{n}$ possible outcomes of just anti-bunched photons, which would make computing the entire probability distribution and then sampling still computationally hard.~\cite{aaronson2014} or~\cite{neville2017} provide a more concrete argument for why sampling from distinguishable photons is classically easy, by demonstrating that we can simply sample from each photon in turn. Alternatively, you might be able to give an argument combining the permanent of positive matrix approximation algorithm with Clifford and Clifford, to demonstrate an easy to simulate algorithm?}

The discussion up to this point shows that results on the classical simulation of the Schur transform would allow us to answer this question, but general results along these lines are to the best of our knowledge not available.
As mentioned above, one way to approach the question is to consider the multipartite entanglement properties of the mixed state of the System that results after tracing the Label.
%States with relatively little entanglement are generally considered easy to simulate.

Without a specific noise model, there are several mixed states we could consider; an obvious one is a mixture of the ideal and completely indistinguishable states
\begin{equation}\label{eq:Werner}
\rho_\epsilon = \epsilon\left(\frac{\sum_{\sigma \in  \textrm{S}_n}\sigma|s\rangle}{\sqrt{n!}}\right)\left(\frac{\sum_{\tau \in  \textrm{S}_n}\langle s|\tau^\dagger}{\sqrt{n!}}\right) + (1-\epsilon)\left(\frac{\sum_{\sigma \in  \textrm{S}_n}\sigma|s\rangle\langle s|\sigma^\dagger}{{n!}}\right).
\end{equation}
Equation~(\ref{eq:ComDis}) tells us that the completely indistinguishable state is maximally mixed on the coincident subspace, and it has been shown that states of the form $\epsilon\ket{\psi}\bra{\psi} + (1-\epsilon)\Id/d$ are separable for sufficiently small $\epsilon$~\cite{braunstein1999, rungta2001}, where $\Id/d$ is the completely mixed state on the entire space (the tensor product $(\mathbb{C}^m)^{\otimes n}$).
It is therefore tempting to conclude that there is a measurable set of states near the completely indistinguishable state that are separable.
However, the completely indistinguishable state only has support on the coincident subspace, which is not a tensor product, and so these results cannot be applied directly.

We can in fact show that for any $\epsilon>0$ the reduced System state of Eq.(\ref{eq:Werner}) is entangled, in that it fails the partial transpose criterion~\cite{chen2002}.
This is similar to results that show a mixture of the totally antisymmetric state and the projector on the symmetric subspace are entangled for two qudits~\cite{breuer2006}.
Rewrite $\rho_\epsilon$ in the form
\begin{align}
\rho_\epsilon
 &= \frac{1}{n!}\left((1-\epsilon)\sum_{\sigma\in \textrm{S}_n}\sigma|s\rangle\langle s|\sigma^\dagger + \epsilon\sum_{\sigma,\tau \in \textrm{S}_n}\sigma|s\rangle\langle s|\tau^\dagger\right)\\
 &= \frac{1}{n!}\left((1-\epsilon)\sum_{\sigma\in \textrm{S}_n}\sigma|s\rangle\langle s|\sigma^\dagger + \epsilon\sum_{\substack{\sigma,\tau \in \textrm{S}_n\\\sigma^{-1}(1) = \tau^{-1}(1)}}\sigma|s\rangle\langle s|\tau^\dagger + \epsilon\sum_{\substack{\sigma,\tau \in \textrm{S}_n\\\sigma^{-1}(1) \neq \tau^{-1}(1)}}\sigma|s\rangle\langle s|\tau^\dagger\right) ,
\end{align}
where we have separated the sum over $\sigma$ and $\tau$ into two sums, based on whether or not $\sigma^{-1}(1) = \tau^{-1}(1)$. 
Transposing the first qudit will leave the first of these sums invariant, while always affecting the second. 
The resulting state is
\begin{align}
\rho_{\epsilon}^{T_1}=
 & \frac{1}{n!} \left((1-\epsilon)\sum_{\sigma\in \textrm{S}_n}\sigma|s\rangle\langle s|\sigma^\dagger + \epsilon\sum_{\substack{\sigma,\tau \in \textrm{S}_n\\\sigma^{-1}(1) = \tau^{-1}(1)}}\sigma|s\rangle\langle s|\tau^\dagger\right.\\
 &\left. + \epsilon\sum_{\substack{\sigma,\tau \in \textrm{S}_n\\\sigma^{-1}(1) \neq \tau^{-1}(1)}}|s_{\tau^{-1}(1)}\rangle\langle s_{\sigma^{-1}(1)}|\bigotimes_{i=2}^n|s_{\sigma^{-1}(i)}\rangle\langle s_{\tau^{-1}(i)}|\right)
\end{align}
To work out the trace norm of this density matrix, we need to multiply it by its transpose. 
This gives us
\begin{align}
\rho_{\epsilon}^{T_1}\rho_{\epsilon}^{T_1\dagger}=
 &\frac{1}{(n!)^2} \left((1-\epsilon)^2\sum_{\sigma\in \textrm{S}_n}\sigma|s\rangle\langle s|\sigma^\dagger + 2(1-\epsilon)\epsilon\sum_{\substack{\sigma,\tau \in \textrm{S}_n\\\sigma^{-1}(1) = \tau^{-1}(1)}}\sigma|s\rangle\langle s|\tau^\dagger\right.\\
&\left.+ \epsilon^2\sum_{\substack{\sigma,\tau,\upsilon,\chi \in \textrm{S}_n\\\sigma^{-1}(1) = \tau^{-1}(1)\\\upsilon^{-1}(1) = \chi^{-1}(1)}}\sigma|s\rangle\langle s|\tau^\dagger\chi|s\rangle\langle s|\upsilon^\dagger\right.\\
&\left. + \epsilon\sum_{\substack{\sigma,\tau,\upsilon,\chi \in \textrm{S}_n\\\sigma^{-1}(1) \neq \tau^{-1}(1)\\\upsilon^{-1}(1) \neq \chi^{-1}(1)}}|s_{\tau^{-1}(1)}\rangle\langle s_{\sigma^{-1}(1)}|s_{\upsilon^{-1}(1)}\rangle\langle s_{\chi^{-1}(1)}|\bigotimes_{i=2}^n|s_{\sigma^{-1}(i)}\rangle\langle s_{\tau^{-1}(i)}|s_{\chi^{-1}(i)}\rangle\langle s_{\upsilon^{-1}(i)}|\right)\\
&= \frac{1}{(n!)^2}\left((1-\epsilon)^2\sum_{\sigma\in \textrm{S}_n}\sigma|s\rangle\langle s|\sigma^\dagger + 2(1-\epsilon)\epsilon\sum_{\substack{\sigma,\tau \in \textrm{S}_n\\\sigma^{-1}(1) = \tau^{-1}(1)}}\sigma|s\rangle\langle s|\tau^\dagger\right.\\
&\left.+ \epsilon^2(n-1)!\sum_{\substack{\sigma,\tau \in \textrm{S}_n\\\sigma^{-1}(1) = \tau^{-1}(1)}}\sigma|s\rangle\langle s|\tau^\dagger\right.
\left. + \epsilon\sum_{\substack{\sigma,\tau,\upsilon \in \textrm{S}_n\\\sigma^{-1}(1) \neq \tau^{-1}(1)\\\sigma^{-1}(1) = \upsilon^{-1}(1)}}|s_{\tau^{-1}(1)}\rangle\langle s_{\tau^{-1}(1)}|\bigotimes_{i=2}^n|s_{\sigma^{-1}(i)}\rangle\langle s_{\upsilon^{-1}(i)}|\right).
\end{align}
The square root of this matrix is
\begin{align}
\sqrt{\rho_\epsilon^{T_1}\rho_\epsilon^{T_1\dagger}}
 &= \frac{1}{n!}\left((1-\epsilon)\sum_{\sigma\in \textrm{S}_n}\sigma|s\rangle\langle s|\sigma^\dagger + \epsilon\sum_{\substack{\sigma,\tau \in \textrm{S}_n\\\sigma^{-1}(1) = \tau^{-1}(1)}}\sigma|s\rangle\langle s|\tau^\dagger\right.\\
&\left. + \frac{\epsilon}{(n-1)!}\sum_{\substack{\sigma,\tau,\upsilon \in \textrm{S}_n\\\sigma^{-1}(1) \neq \tau^{-1}(1)\\\sigma^{-1}(1) = \upsilon^{-1}(1)}}|s_{\tau^{-1}(1)}\rangle\langle s_{\tau^{-1}(1)}|\bigotimes_{i=2}^n|s_{\sigma^{-1}(i)}\rangle\langle s_{\upsilon^{-1}(i)}|\right).\label{eqn:mismatched-perms}
\end{align}
From this we can work out the trace norm as
\begin{align}
\|\rho_{\epsilon}^{T_1}\|_*
 &= \textrm{Tr}\left[\sqrt{\rho_\epsilon^{T_1}\rho_\epsilon^{T_1\dagger}}\right]\\
 &= \frac{1}{n!}\left((1-\epsilon)n! + \epsilon n! + \frac{\epsilon}{(n-1)!}(n-1)n!(n-1)!\right)\\
 &= 1-\epsilon + \epsilon+\epsilon(n-1)\\
 &= 1+\epsilon(n-1).
\end{align}
Here, the $(1-\epsilon)n!$ and $\epsilon n!$ terms come from the trace of the first two terms. 
The third term comes from Eq.(\ref{eqn:mismatched-perms}), the trace of which one can think of as the number of ways we can pick permutations $\sigma,\tau\in\textrm{S}_n$ such that $\sigma^{-1}(1)\neq\tau^{-1}(1)$. 
This can be worked out by choosing any permutation $\sigma \in \textrm{S}_n$, and constructing $\tau$ by first choosing $\tau^{-1}(1)\neq\sigma^{-1}(1)$ and choosing $\tau^{-1}(i)$ for $i>1$ to be a permutation in $\textrm{S}_{n-1}$. 
Thus the overall number is $(n-1)n!(n-1)!$.

We can see that the trace norm of $\rho_\epsilon^{T_1}$ therefore fails the generalised partial transposition criterion for separability if $n>1$ and $\epsilon>0$.
Note that this does not imply that efficiently sampling from unitary actions on such states is not classically possible for any nonzero $\epsilon$, merely that techniques used for simulating separable states cannot be used for exact sampling in this case.

\section{Conclusion}
\label{sec:conclusion}

We have described how to use the Schur transform to perform a quantum simulation of bosonic sampling when the bosons are arbitrarily distinguishable.
These results make it clear that ideal $n$ boson, $m$ mode linear interferometry is equivalent to a transversal $n$ qudit quantum circuit, with the constraint that the input must be totally symmetric -- that is, the ordering of the qudits must be erased.
Moreover, we can introduce nonideal aspects into the quantum simulation by tracing out qudits (loss), or introducing ancillas and entanglement (distinguishability).
A recently released paper focusing on the issue of loss in more detail makes similar connections~\cite{oszmaniec2018}.

A broad aim of future research along this approach is to better understand how the computational complexity of sampling from photons changes as photons become more distinguishable. 
%As stated in Section \ref{sec:intro}, currently such results are only known for the perfectly distinguishable and indistinguishable cases. 
By better understanding of how these intermediate levels of distinguishability link to representation theory, the hope is that it will be easier to find either classical algorithms for these cases similarly to Clifford and Clifford~\cite{clifford2017}, or in finding reductions to sampling imminants~\cite{mertens2013} much as Aaronson and Arkhipov focused on permanents. 
Indeed, recent work by Havl\'{i}\u{c}ek and Strelchuk has demonstrated the potential for the use of the Schur transform in understanding the complexity of sampling problems~\cite{havlicek2018}.

\section*{Acknowledgements}
AEM was supported by the Bristol Quantum Engineering Centre for Doctoral Training, EPSRC grant EP/L015730/1. 
PST was supported in part by an EPSRC First Grant EP/N014812/1. 
We would like to thank A Harrow and S Stanisic for insightful discussions, and A Montanaro for helpful comments.
No new data were created during this study.

\bibliographystyle{unsrt}
\bibliography{noisy_circuit}

\end{document}